\documentclass{article}
\usepackage{amsmath,amsthm,amssymb}
\usepackage{authblk}
\usepackage{fullpage}
\title{Optimal Message-Passing with Noisy Beeps}
\author{Peter Davies}
\affil{Durham University}

\usepackage{xspace}

\usepackage{graphicx}
\usepackage{cleveref}
\newtheorem{theorem}{Theorem}
\newtheorem{corollary}[theorem]{Corollary}
\newtheorem{lemma}[theorem]{Lemma}

\newtheorem{definition}[theorem]{Definition}
\newtheorem{notation}[theorem]{Notation}

\newcommand{\Prob}[1]{\mathbf{Pr}\left[#1\right]}
\newcommand{\Exp}[1]{\mathbf{E}\left[#1\right]}
\newcommand{\nat}{\ensuremath{\mathbb{N}}}
\newcommand{\congest}{\textsf{CONGEST}\xspace}
\newcommand{\bcongest}{\textsf{Broadcast CONGEST}\xspace}
\def\epsilon{\ensuremath{\varepsilon }}
\newcommand{\eps}{\ensuremath{\epsilon }}

\usepackage{algorithm}
\usepackage{algpseudocode}

\newcommand{\COMMENTED}[1]{{}}
\newcommand{\hide}[1]{\COMMENTED{#1}}

\date{}
\begin{document}
\maketitle

\begin{abstract}
Beeping models are models for networks of weak devices, such as sensor networks or biological networks. In these networks, nodes are allowed to communicate only via emitting beeps: unary pulses of energy. Listening nodes only the capability of {\it carrier sensing}: they can only distinguish between the presence or absence of a beep, but receive no other information. The noisy beeping model further assumes listening nodes may be disrupted by random noise.

Despite this extremely restrictive communication model, it transpires that complex distributed tasks can still be performed by such networks. In this paper we provide an optimal procedure for simulating general message passing in the beeping and noisy beeping models. We show that a round of \textsf{Broadcast CONGEST} can be simulated in $O(\Delta\log n)$ round of the noisy (or noiseless) beeping model, and a round of \textsf{CONGEST} can be simulated in $O(\Delta^2\log n)$ rounds (where $\Delta$ is the maximum degree of the network). We also prove lower bounds demonstrating that no simulation can use asymptotically fewer rounds.

This allows a host of graph algorithms to be efficiently implemented in beeping models. As an example, we present an $O(\log n)$-round \textsf{Broadcast CONGEST} algorithm for maximal matching, which, when simulated using our method, immediately implies a near-optimal $O(\Delta \log^2 n)$-round maximal matching algorithm in the noisy beeping model.

\end{abstract}

\section{Introduction}
Beeping models were first introduced by Cornejo and Kuhn \cite{CK10} to model wireless networks of weak devices, such as sensor networks and biological networks \cite{Science}. These models are characterised by their very weak assumptions of communication capabilities: devices are assumed to communicate only via \emph{carrier sensing}. That is, they have the ability to distinguish between the presence or absence of a signal, but not to gain any more information from the signal. 

\subsection{Models}
The models we study all have the same basic structure: a network of devices is modeled as a graph with $n$ nodes (representing the devices) and maximum degree $\Delta$, where edges represent direct reachability between pairs of devices. We will assume that all nodes activate simultaneously, and therefore have shared global clock (some prior work on beeping models instead allow nodes to activate asynchronously). Time then proceeds in synchronous rounds, in which nodes can perform some local computation and then can communicate with neighboring devices. The defining characteristic of each model is the communication capability of the nodes.

\paragraph{Noiseless Beeping Model}
In each round, each node chooses to either beep or listen. Listening nodes then hear a beep iff at least one of their neighbors beeped, and silence otherwise. Nodes do not receive any other information about the number or identities of their beeping neighbors.

\paragraph{Noisy Beeping Model}
The noisy beeping model, introduced by Ashkenazi, Gelles, and Leshem \cite{AGL20}, is similar to the noiseless version, except that the signal each listening node hears (beep or silence) is \emph{flipped}, independently uniformly at random, with some probability $\eps\in(0,\frac 12)$.

Within these beeping models, our aim will be to simulate more powerful message-passing models, in which nodes have the ability to send longer messages to each other, and these messages are received without interference:

\paragraph{\bcongest Model}
In rounds of the \bcongest model, nodes may send the same $O(\log n)$-bit message to each of their neighboring nodes, and each node hears the messages from all of its neighbors.

\paragraph{\congest Model}
The \congest model is similar to \bcongest, but allows nodes to send (potentially) different $O(\log n)$-bit messages to each of their neighboring nodes. Again, each node hears the messages from all of its neighbors.

The communication capabilities in the \bcongest and \congest models are clearly much more powerful than that of either beeping model, and \congest in particular has a broad literature of efficient algorithms. Our aim in this work is to provide an efficient generic simulation of \bcongest and \congest in the beeping models, so that these existing algorithms can be applied out-of-the-box to networks of weak devices.

\subsection{Prior work}

\paragraph{Beeping models}
The (noiseless) beeping model was introduced by Cornejo and Kuhn \cite{CK10}, who also gave results for an interval coloring task used for synchronization. Classical local graph problems have been studied in the model, with Afek et al. \cite{AABCHK13} giving an $O(\log^2 n)$-round maximal independent set algorithm, and Beauqier et al. \cite{BBDK18} giving $O(\Delta^2 \log n + \Delta^3)$-round deterministic algorithms for maximal independent set and $(\Delta+1)$-coloring. 

Global communication problems (those requiring coordination across the entire network, and therefore with running times parameterized by the diameter $D$ of the network) have also been studied. Single-source broadcast of a $b$-bit message can be performed in $O(D+b)$ rounds using the simple tool of `beep waves', introduced by Ghaffari and Haeupler \cite{GH13} and formalized by Czumaj and Davies \cite{CD19}. Leader election, another fundamental global problem, has seen significant study in the model.  Ghaffari and Haeupler \cite{GH13} gave a randomized algorithm requiring $O(D+\log n \log \log n)\cdot \min\{\log\log n, \log\frac nD\}$, while F{\"o}rster, Seidel and Wattenhofer \cite{FSW14} gave an $O(D \log n)$-round \emph{deterministic} algorithm. Czumaj and Davies \cite{CDleader} gave a simple randomized algorithm with $O(D \log n)$ worst-case round complexity but $O(D+\log n)$ expected complexity. Finally Dufoulon, Burman and Beauquier \cite{DBB18} settled the complexity of the problem with a deterministic algorithm with optimal $O(D+\log n)$ round complexity.

On other global problems, Czumaj and Davies \cite{CD19} and Beauqier et al. \cite{BBDD19} gave results for broadcasting from multiple sources, and Dufoulon, Burman and Beauqier \cite{DBB20} study synchronization primitives for the model variant where nodes activate asynchronously. 

\paragraph{Message passing models}

Message passing models, and \congest in particular, have seen a long history of study and have a rich literature of algorithms for problems including (among many others) local problems such as $\Delta+1$-coloring\cite{HKMT21}, global problems such as minimum spanning tree\cite{KKT15}, and approximation problems such as approximate maximum matching\cite{AKO18}. \bcongest is less well-studied, though some dedicated algorithms have also been developed for it, e.g. \cite{HP21}. There is an obvious way to simulate \congest algorithms in \bcongest at an $O(\Delta)$-factor overhead: nodes simply broadcast the messages for each of their neighbors in turn, appending the ID of the intended recipient. In general this is the best that can be done (as can be seen from our bounds on simulating beeping models), but for specific problems this $\Theta(\Delta)$ complexity gap is often not necessary.

\paragraph{Simulating message passing with beeps}
Two works have previously addressed the task of simulating messaging passing in beeping models. The first was by Beauquier et al. \cite{BBDK18}, and gave a generic simulation for \congest in the noiseless beeping model. Their algorithm required $\Delta^6$ setup rounds, and then $\Delta^4\log n$ beep-model rounds per round of \congest. This result was improved by Ashkenazi, Gelles, and Leshem \cite{AGL20}, who introduced the noisy beeping model, and gave an improved simulation of \congest which requires $O(\Delta^4\log n)$ rounds of setup, and then simulates each \congest round in $O(\Delta\log n \cdot \min\{n,\Delta^2\})$ rounds of noisy beeps. 

\subsection{Our results}
We give a randomized simulation of \bcongest which requires $O(\Delta\log n)$ rounds in the noisy beep model per round of \bcongest, with no additional setup cost. We will call this per-round cost the \emph{overhead} of simulation. This implies a simulation of \congest with $O(\Delta^2\log n)$ overhead in the noisy beep model. We therefore improve over the previous best result of \cite{AGL20} by reducing the overhead by a $\Theta(\min\{\frac{n}{\Delta},\Delta\})$ factor, and removing the large setup cost entirely. We prove that these bounds are tight for both \bcongest and \congest by giving matching lower bounds (even for the noiseless beeping model). This has the potentially surprising implication that introducing noise into the beeping model does not asymptotically increase the complexity of message-passing simulation at all.

This simulation result allows many \congest and \bcongest algorithms to be efficiently implemented with beeps. As an example, we show an $O(\log n)$-round \bcongest algorithm for the task of maximal matching, which via our simulation implies an $O(\Delta\log^2 n)$-round algorithm in the noisy beeping model. We show that this is almost optimal by demonstrating an $\Omega(\Delta \log n)$ lower bound (even in the noiseless model).

\subsection{Our Approach}
We summarize our approach to simulating \congest in the noiseless beeping model (the noisy case will follow naturally, as we will see later). First, let us mention the general approach of the previous results of \cite{BBDK18} and \cite{AGL20}: there, the authors use a coloring of $G^2$ (i.e., a coloring such that no nodes within distance $2$ in $G$ receive the same color) to sequence transmissions. They iterate through the color classes, with nodes in each class transmitting their message (over a series of rounds, with a beep or silence representing each bit of the message). Since nodes have at most one neighbor in each color class, they hear that neighbor's message undisrupted.

The disadvantage of such an approach is that the coloring of $G^2$ requires a large setup time to compute, and also necessitates at least $\min\{n,\Delta^2\}$ color classes. This is the cause of the larger overhead in the simulation result of \cite{AGL20}. 

Instead of having nodes transmitting at different times, our solution is to have them all transmit at once, and use superimposed codes to ensure that the messages are decipherable. The definition of a classic superimposed code is as follows:

\begin{definition}[Superimposed Codes]
An $(a,k)$-superimposed code of length $b$ is a function $C:\{0,1\}^a\rightarrow\{0,1\}^b$ such that any superimposition (bitwise OR) of at most $k$ codewords is unique.
\end{definition}

The connection between superimposed codes and beeping networks is that, if some subset of a node $v$'s neighbors all transmit a message simultaneously (using beeps to represent $\textbf 1$s and silence to represent $\textbf 0$s), then $v$ (if it were to listen every round) would hear the bitwise OR superimposition of all the messages. If this superimposition is unique, then $v$ is able to identify the set of messages that were transmitted (and this set contain precisely those messages with no $\textbf 1$ in a position where the superimposition has $\textbf 0$).

Superimposed codes of this form were first introduced by Kautz and Singleton \cite{KS64}, who showed a construction with $b=O(k^2  a)$. This definition is equivalent to cover-free families of sets, which is the terminology used in much of the prior work. A lower bound $b=\Omega(\frac{k^2 a}{\log k} )$ was found by D'yachkov and Rykov \cite{DR82}, with a combinatorial proof later given by Ruszink{\'o} \cite{R94}, and another, simple proof given by F{\"u}redi \cite{Furedi96}. The $\log k$ gap between upper and lower bounds remains open.

This presents a problem to applying such codes for message passing in the beep model. If all nodes are transmitting their message (of $O(\log n)$ bits) at once, then we would need to use an $(O(\log n),\Delta)$-superimposed code for the messages to be decodable. Using Kautz and Singleton's construction \cite{KS64} results in a length of $O(\Delta^2\log n)$ (and length corresponds directly to rounds in the beeping model). This would result in the same $O(\Delta^2)$-factor overhead as from using a coloring of $G^2$, so would not improve over \cite{AGL20}. Furthermore, even if we were to find improved superimposed codes, the lower bound implies that any such improvement would be only minor.

To achieve codes with better length, we weaken the condition we require. Rather than requiring that all superimpositions of at most $k$ codewords are unique, we only require that \emph{most} are. Specifically, if the $k$ codewords are chosen at random, then their superimposition will be unique (and hence decodable) with high probability. We show the existence of short codes with this weakened property. Constructions with similar properties (though not quite suitable for our uses) were also given in \cite{DVPS17}.

This raises a new problem: using these shorter codes, we can efficiently have all nodes send a \emph{random} message to their neighbors, but how does this help us send a specific message?

Our answer is that if we repeat the transmission (using the same random codewords for each node), then every node $v$ already knows exactly when its neighbors should be beeping\footnote{Technically, $v$ does not know which neighbor corresponds to which codeword, but this is not required by our approach.}, and in particular, $v$ knows when a neighbor $u$ should be beeping \emph{alone} (i.e., not at the same time as any other neighbor of $v$). If $u$ now beeps only in a \emph{subset} of the rounds indicated by its codeword, then it can pass information to $v$ in this way. So, our final algorithm uses a secondary \emph{distance} code to specify what this subset should be in order to ensure that all neighbors of $u$ can determine $u$'s message. The aim of this distance code is that codewords are sufficiently large Hamming distance apart that $u$'s neighbors can determine $u$'s message, even though they only hear a subset of the relevant bits, and these bits can be flipped by noise in the noisy model.

\subsection{Notation}
Our protocols will be heavily based on particular types of binary codes, which we will communicate in the beeping model via beeps and silence. In a particular round, in the noiseless beeping model, we will say that a node $v$ receives a $\mathbf{1}$ if it either listens and hears a beep, or beeps itself. We will say that $v$ receives a $\mathbf{0}$ otherwise. In the noisy model, what $v$ hears will be this bit, flipped with probability $\eps$.

We will use logic operators to denote operations between two strings: for $s,s'\in \{0,1\}^a$, $s\land s'\in\{0,1\}^a$ is the logical \textsc{And} of the two strings, with  $\mathbf{1}$ in each coordinate iff both $s$ and $s'$ had $\mathbf{1}$ in that coordinate. Similarly, $s\lor s'\in\{0,1\}^a$ is the logical \textsc{Or} of the two strings, with  $\mathbf{1}$ in each coordinate iff $s$ or $s'$ (or both) had $\mathbf{1}$ in that coordinate. 

\begin{definition}
We will use $\mathbf 1(s)$ to denote the number of $\mathbf 1$s in a string $s\in \{0,1\}^a$. We will say that a string $s\in \{0,1\}^a$ $d$-intersects another string $s'\in \{0,1\}^a$ if $\mathbf{1}(s\land s')\ge d$. 
\end{definition}

For a set of strings $S\in \{0,1\}^a$, we will use $\vee(S)$ as shorthand for the superimposition $\bigvee_{s\in S}s$.
\section{Binary Codes}

The novel type of superimposed code on which our algorithm is mainly based is defined as follows:
 
\begin{definition}
An $(a,k,\delta)$-beep code of length $b$ is a function $C:\{0,1\}^a\rightarrow\{0,1\}^b$ such that:

\begin{itemize}
	\item all $s\in C$ have $\mathbf{1}(s)=\frac{\delta b}{k}$.
	\item the number of size-$k$ subsets $S\subseteq C$ whose superimpositions $\vee(S)$ $\frac{5\delta^2 b}{k}$-intersect some $s\in C\setminus S$ is at most $\binom{2^a}{k}2^{-2a}$
\end{itemize}

(here we slightly abuse notation by using $C$ to denote the set of codewords, i.e. the image $C(\{0,1\}^a)$ of the beep code function).
\end{definition}

In other words, all codewords have exactly $\frac{\delta b}{k}$ $\mathbf{1}$s, and only a $2^{-2a}$-fraction of the $\binom{2^a}{k}$ size-$k$ subsets of codewords have a superimposition that $\frac{5\delta^2 b}{k}$-intersects some other codeword. This first criterion is only a technicality to aid our subsequent application; the important point is the second, which will imply that a superimposition of $k$ \emph{random} codewords will, with probability at least $1-2^{-2a}$, be decodable (even under noise, since avoiding $\frac{5\delta^2 b}{k}$-intersection will provide use with sufficient redundancy to be robust to noise). Note that for such a code to exist, $\frac{\delta b}{k}$ must be an integer, which we will guarantee in our construction.

\begin{theorem}\label{thm:beepcode}
	For any any $a,k,c\in \nat$, there exists an $(a,k,1/c)$-beep code of length $b=c^2 ka$.
\end{theorem}

\hide{
\begin{proof}
	The proof will be by the probabilistic method: we will randomly generate a candidate code $C$, and then prove that it has the desired properties with high probability in $2^a$. Then, a code with such properties must exist, and the random generation process we use implies an efficient algorithm to find such a code with high probability (though we cannot \emph{check} this efficiently).
	
	To generate our candidate code, we choose each codeword independently, uniformly at random from the set of all $b$-bit strings with $\frac{ b}{ck}$ $\mathbf{1}$s. This clearly guarantees the first property.
	
	For a fixed size-$k$ set $S$ of codewords, and a fixed codeword $x\in C\setminus S$, we now analyze the probability that $\vee(S)$ $\frac{5 b}{c^2k}$-intersects $s$.
	
	Clearly we have $\mathbf 1(\vee(S)) \le k\cdot \frac{ b}{ck} =  b/c$. Fixing the random choices of the codewords in $S$ and considering the choice of $s$, we see that there are at most $\binom{ b/c}{\frac{ b}{3ck}}$ ways of picking $\frac{ b}{3ck}$ $\mathbf 1$s that intersect $\vee(S)$, and at most $\binom{ b}{\frac{2 b}{3ck}}$ ways of picking $s$'s remaining $\frac{2 b}{3ck}$ $\mathbf 1$s, compared to $\binom{ b}{\frac{ b}{ck}}$ total choices for choosing $s$. So,

	\[
	\Prob{\text{$\vee(S)$ $\frac{ b}{3ck}$-intersects $s$}}\le
	\binom{ b/c}{\frac{ b}{3ck}}\cdot\binom{ b}{\frac{2 b}{3ck}}/ \binom{ b}{\frac{ b}{ck}}\enspace.\]

	Using the well-known bounds on binomial coefficients $(\frac{x}{y})^y \le \binom xy \le (\frac{ex}{y})^y$, we get the following:
	\begin{align*}
		\Prob{\text{$\vee(S)$ $\frac{ b}{3ck}$-intersects $s$}}
		&\le
		\left(3ek\right)^{\frac{ b}{3ck}} 
		\cdot \left(\frac{3eck}{2}\right)^{\frac{2 b}{3ck}}
		/\left(ck\right)^{\frac{b}{ck}}\\
		&=
		(3e)^{\frac{ b}{ck}}2^{\frac{-2 b}{3ck}}c^{\frac{ -b}{3ck}}\\
		&=
		\left(\frac{27e^3}{4c}\right)^{\frac{ b}{3ck}}\\
		&<\left(0.91\right)^{\beta a/3}\\
		&< 2^{-4a}
		\enspace.
	\end{align*}
	
	Taking a union bound over all codewords $s\in C\setminus S$, we find that the probability that $S$ $\frac{ b}{3ck}$-intersects any such codeword is at most $2^{-3a}$. Then, the expected number of size-$k$ sets $S$ that $\frac{ b}{3ck}$-intersect any $s\in C\setminus S$ is at most $\binom{2^a}{k} 2^{-3a}$. By the probabilistic method, there therefore \emph{exists} a an $(a,k,1/c)$-beep code in which the number of size-$k$ sets $S$ that $\frac{ b}{3ck}$-intersect any $s\in C\setminus S$ is at most $\binom{2^a}{k} 2^{-3a}$. 
	
	However, since we also want an efficient algorithm to \emph{find} an $(a,k,1/c)$-beep code, we note that by Markov's inequality the probability that more than $\binom{2^a}{k} 2^{-2a} $ size-$k$ sets $S$ that $\frac{ b}{3ck}$-intersect any $s\in C\setminus S$ is at most $2^{-a}$, and therefore the process of choosing codewords uniformly at random from all strings with $\frac{ b}{ck}$ $\mathbf{1}$s gives an $(a,k,1/c)$-beep code with probability at least $1-2^{-a}$.
\end{proof}}

\begin{proof}
	The proof will be by the probabilistic method: we will randomly generate a candidate code $C$, and then prove that it has the desired properties with high probability in $2^a$. Then, a code with such properties must exist, and the random generation process we use implies an efficient algorithm to find such a code with high probability (though \emph{checking} the code is correct would require $2^{O(ak)}$ computation).
	
	To generate our candidate code, we choose each codeword independently, uniformly at random from the set of all $b$-bit strings with $\frac{ b}{ck}$ $\mathbf{1}$s. This clearly guarantees the first property.
	
	For a fixed size-$k$ set $S$ of codewords, and a fixed codeword $x\in C\setminus S$, we now analyze the probability that $\vee(S)$ $\frac{5 b}{c^2k}$-intersects $x$.
	
	Clearly we have $\mathbf 1(\vee(S)) \le k\cdot \frac{ b}{ck} =  b/c$. Consider the process of randomly choosing the positions of the \textbf 1s of $x$. Each falls in the same position as a \textbf 1 of $\vee(S)$ with probability at most $1/c$, even independently of the random choices for the other \textbf{1}s. The probability that $\vee(S)$ $\frac{5 b}{c^2k}$-intersects $x$ is therefore at most \[\binom{\frac{b}{ck}}{\frac{5 b}{c^2k}} \cdot c^{-\frac{5 b}{c^2k}}
	\le \left(\frac{ec}{5}\right)^{\frac{5 b}{c^2k}}\cdot c^{-\frac{5 b}{c^2k}}
	\le \left(\frac{5}{e}\right)^{-\frac{5  c^2ka}{c^2k}}
	\le 2^{-4  a}
	\]
	
	Taking a union bound over all codewords $s\in C\setminus S$, we find that the probability that $\vee(S)$ $\frac{5 b}{c^2k}$-intersects any such codeword is at most $2^{-4a}$. Then, the expected number of size-$k$ sets $S$ that $\frac{5 b}{c^2k}$-intersect any $s\in C\setminus S$ is at most $\binom{2^a}{k} 2^{-3a}$. By the probabilistic method, there therefore \emph{exists} a an $(a,k,1/c)$-beep code in which the number of size-$k$ sets $S$ that $\frac{5 b}{c^2k}$-intersect any $s\in C\setminus S$ is at most $\binom{2^a}{k} 2^{-3a}$. 
	
	However, since we also want an efficient algorithm to \emph{find} an $(a,k,1/c)$-beep code, we note that by Markov's inequality the probability that more than $\binom{2^a}{k} 2^{-2a} $ size-$k$ sets $S$ that $\frac{5 b}{c^2k}$-intersect any $s\in C\setminus S$ is at most $2^{-a}$, and therefore the process of choosing codewords uniformly at random from all strings with $\frac{ b}{ck}$ $\mathbf{1}$s gives an $(a,k,1/c)$-beep code with probability at least $1-2^{-a}$.

\end{proof}

	Notice that, while the theorem holds for any $c\in \nat$, it is trivial for $c\le 2$: in this case, codewords cannot $\frac{5 b}{c^2k}$-intersect any string, since they contain only $\frac{b}{ck}$ \textbf 1s. Our application will set $c$ to be a sufficiently large constant.

Our algorithm will also make use of \emph{distance codes}. These codes have the simple criterion that every pair of codewords is sufficiently far apart by Hamming distance (which we will denote $d_H$). Distance codes are an example of error-correcting codes, which have a wealth of prior research (see e.g. \cite{HP10} for an extensive survey); here we just require a very simple object, for which we give a proof in a similar style to that of \Cref{thm:beepcode} for consistency:

\begin{definition}
An $(a,\delta)$-distance code of length $b$ is a function $D:\{0,1\}^a\rightarrow\{0,1\}^b$ such that all pairs $s\ne s'\in D$ have $d_H(s,s')\ge \delta b$.
\end{definition}

\begin{lemma}\label{lem:distcode}
	For any $\delta\in(0,\frac 12)$, $a\in \nat$, and $c_\delta \ge 12(1-2\delta)^{-2}$, there exists an $(a,\delta)$-distance code of length $b=c_\delta a$.
\end{lemma}

\begin{proof}
	We randomly generate a candidate code by choosing each codeword's entries independently uniformly at random from $\{0,1\}$. For any pair of codewords $s,s'\in D$, the probability that they differ on any particular entry is $\frac 12$. The expected distance is therefore $\frac b2$, and by a Chernoff bound, 
	
	\begin{align*}
		\Prob{d_H(s,s') \le \delta b} &= \Prob{d_H(s,s')\le 2\delta \Exp{d_H(s,s')}}\\
		&\le e^{\frac{-(1-2\delta)^2\Exp{d_H(s,s')}}{2}}
		= 	 e^{\frac{-(1-2\delta)^2c_\delta a}{4}}\enspace.
	\end{align*}

Since $c_\delta \ge 12(1-2\delta)^{-2}$,

	\begin{align*}
	\Prob{dist(s,s') \le \delta b} &\le 	 e^{-3 a}\le 2^{-4a}\enspace.
\end{align*}
	
	Taking a union bound over all $\binom{2^a}{2}\le 2^{2a}$ pairs $s,s'\in D$, we find that the probability that any pair has $dist(s,s')\le \delta b$ is at most $2^{-2a}$. Therefore, the random generation process generates an $(a,\delta)$-distance code with probability at least $1-2^{-2a}$.
\end{proof}

This construction can also be checked relatively efficiently, since one need only check the distance of $O(2^{2a})$ codeword pairs, which can be performed in $2^{O(a)}$ computation.

\section{Simulation Algorithm}
We now arrive at our main simulation algorithm. We give an algorithm for simulating a single communication round in \bcongest using $O(\Delta\log n)$ rounds of the noisy beep model. What we mean by this simulation is that each node $v$ begins with a $\gamma \log n$-bit message $m_v$ to transmit to all neighbors (where $\gamma$ is some constant, a parameter of the \bcongest model), and by the end of our beeping procedure, all nodes should be able to output the messages of all their neighbors.

Let $c_\eps$ be a constant to be chosen based on $\eps$, the noise constant. Our algorithm will make use of two codes (instantiations of those defined in the previous section): 

\begin{itemize}
	\item a $(\gamma\log n,\frac 13)$-distance code $D$ of length $c_\eps^2 \gamma \log n$, given by Lemma \ref{lem:distcode} (so long as we choose $c_\eps\ge 108$); 
	\item a $(c_\eps \gamma \log n,\Delta+1,1/c_\eps)$-beep code $C $ of length $c_\eps^3 \gamma (\Delta+1)\log n$ given by Theorem \ref{thm:beepcode} .

\end{itemize}

 The codewords in the beep code $C$ contain exactly $c_\eps^2 \gamma \log n$ $\mathbf 1$s. The purpose of using these two codes is to combine them in the following manner:

 \begin{notation}
For a binary string $s$, let $\textbf 1_i(s)$ denote the position of the $i^{th}$ $\textbf 1$ in $s$ (and \textsc{Null} if $s$ contains fewer than $i$ $\textbf 1$s).
 \end{notation}
 
 Let $CD:\{0,1\}^{c_\eps \gamma \log n}\times \{0,1\}^{\gamma \log n} \rightarrow \{0,1\}^{c_\eps^3 \gamma (\Delta+1)\log n}$ be the combined code defined as follows:
 
 \[CD(r,m)_j = \begin{cases} \mathbf 1 &\text{if for some $i\in [c_\eps^2 \gamma \log n]$,  $\textbf 1_i(C(r)) = j$, and $D(m)_i = \mathbf 1$}\\
 \mathbf 0	& \text{otherwise}\end{cases}\]
 
That is, $CD(r,m)$ is the code given by writing the codeword $D(m)$ in the positions where $C(r)$ is $\textbf 1$ (and leaving the other positions as $\textbf 0$): see Figure \ref{fig:code}.

\begin{figure}
		\centering
		\includegraphics[width=.9\linewidth]{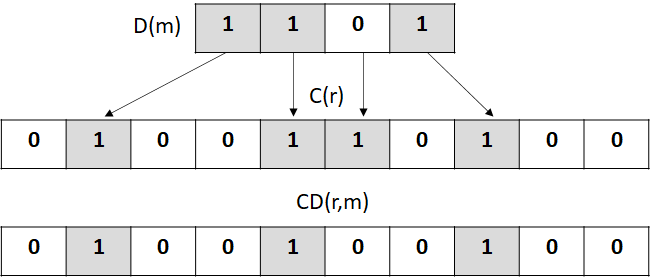}
		\caption{Combined code construction}
		\label{fig:code}
\end{figure}

The algorithm is then as follows (Algorithm \ref{alg:beepsim}):

\begin{algorithm}[H]
 	\caption{Simulation of a \bcongest round in the noisy beeping model}
 	\label{alg:beepsim}
 	\begin{algorithmic}
 		\State Each node $v$ picks $r_v\in\{0,1\}^{c_\eps \gamma \log n}$ independently uniformly at random
 		 \For{$i = 1$ to $c_\eps^3 \gamma (\Delta+1)\log n$, in round $i$,} 			
 		\State Node $v$ beeps iff $C(r_v)_i=1$	
 		\EndFor 	
 		\For{$i = 1$ to $c_\eps^3 \gamma (\Delta+1)\log n$, in round $i+c_\eps^3 \gamma (\Delta+1)\log n$,} 			
			\State Node $v$ beeps iff $CD(r_v,m_v)_i=1$	
		\EndFor 		
 	\end{algorithmic}
 \end{algorithm}

So, each node picks a random codeword from the beep code, and transmits it bitwise using beeps and silence. By the properties of the beep code, with high probability the superimposition of messages each node receives will be decodable. Then, to actually convey the message $m_v$, $v$ uses the combined code, which transmits $m_v$, encoded with a distance code, in the positions where the beep codeword $r_v$ used in the first round was $\textbf 1$. Neighbors $u$ of $v$ know when these positions are from the first round. Of course, there are some rounds when other neighbors of $u$ will be beeping, some rounds when $u$ must beep itself and cannot listen, and some rounds when the signal from $v$ is flipped by noise. However, we will show that, by a combination of the properties of our two codes, there is sufficient redundancy to overcome all three of these obstacles, and allow $u$ to correctly decode $v$'s message.

\section{Decoding the Code}
In the first phase, each node $v$ hears\footnote{For simplicity of notation, we will assume that a node counts a round in which it itself beeped as `hearing' a $\textbf 1$, and in the noisy model, flips this $\textbf 1$ with probability $\eps$ itself. Of course, in practice this is unnecessary, and having full information about its own message can only help a node.} a string we will denote $\tilde x_v$, which is the string $x_v := \bigvee_{u\in N(v)} C(r_u)$ with each bit flipped with probability $\eps\in (0,\frac 12)$, and the aim is for $v$ to decode this string in order to determine the set $R_v:= \{r_u:u\in N(v)\}$. 

We first show that, before considering noise, with high probability the superimposition of random codewords chosen by each node's inclusive neighborhood is decodable. 

\begin{lemma}\label{lem:notintersect}
	With probability at least $1- n^{3-c_\eps \gamma}$, for every node $v\in V$ and every $r\in \{0,1\}^{c_\eps \gamma \log n}$, $C(r)$ does not $5c_\eps \gamma \log n$-intersect $\bigvee_{r_u\in R_v\setminus \{r\}} C(r_u)$.
\end{lemma}

\begin{proof}
First, we see that with probability at least $1-\frac{n^2}{2^{c_\eps \gamma \log n}} = 1-n^{2-c_\eps \gamma}$, all nodes choose different random strings. For the rest of the proof we condition on this event.
	
For each $v\in V$, $r\in \{0,1\}^{c_\eps \gamma \log n}$, let $R_{v,r}$ be a set of nodes' random strings defined as follows: starting with $R_v\setminus \{r\}$ (which is a set of input messages of size at most $\Delta+1$), add arbitrary $r_x$ from nodes $x\notin (N(v)\cup \{w\})$ until the set is of size exactly $\Delta+1$. Since we are conditioning on the event that all nodes generate different random strings, $R_{v,w}$ is a set of $\Delta+1$ distinct random strings from $\Delta+1$ distinct nodes, none of which are $w$.

By the properties of a $(c_\eps \gamma \log n,\Delta+1,1/c_\eps)$-beep code, therefore, the probability that $C(r)$ $5c_\eps \gamma \log n$-intersects $\bigvee_{r_u\in R_{v,w}}  C(r_u)$ is at most $2^{-2c_\eps \gamma \log n} = n^{-2c_\eps \gamma}$. If $C(r)$ does not $5c_\eps \gamma \log n$-intersect $\bigvee_{r_u\in R_{v,w}}  C(r_u)$, then it also does not $5c_\eps \gamma \log n$-intersect $\bigvee_{r_u\in R_v\setminus \{r\}}  C(r_u)$, since $R_{v,w}$ is a superset of $R_v\setminus \{r\}$.

The number of possible pairs $v\in V$, $r\in \{0,1\}^{c_\eps \gamma \log n}$ is $n^{1+c_\eps \gamma}$. Taking a union bound over all of these, we find that $C(r)$ does not $5c_\eps \gamma \log n$-intersect $\bigvee_{r_u\in R_v\setminus \{r\}} C(r_u)$ for any pair with probability at least $1-n^{1+c_\eps \gamma-2c_\eps \gamma} = 1-n^{1-c_\eps \gamma}$ by a union bound. Finally, removing the conditioning on the event that nodes' random strings are all different, we reach the condition of the lemma with probability at least $1-n^{1-c_\eps \gamma} - n^{2-c_\eps \gamma} \ge 1- n^{3-c_\eps \gamma}$.
\end{proof}

Next we must analyze how noise affects the bitstrings that nodes hear. For any node $v$, let $x_v$ denote the string $v$ heard, i.e., $\bigvee_{u\in N(v)} C(r_u)$, after each bit is flipped with probability $\eps\in (0,\frac 12)$. To decode the set $R_v$, $v$ will take $\tilde R_v = \{r\in \{0,1\}^{c_\eps \gamma \log n}: \text{$C(r)$ does not $\frac{2\eps+1}{4}c_\eps^2\gamma \log n$-intersect }\lnot \tilde x_v\}$. That is, it includes all codewords which have fewer than $\frac{2\eps+1}{4}c _\eps^2\gamma \log n$ \textbf 1s in positions where $\tilde x_v$ does not.

Notice that, in the absence of noise, all $C(r)$ for $r\in R_v$ have zero \textbf 1s in positions where $x_v$ did not, and all $C(r)$ for $r\notin R_v$ have at least $c_\eps(c_\eps - 5)\gamma \log n$, since $C(r)$ contains exactly $c_\eps^2 \gamma \log n$ \textbf 1s and, by \Cref{lem:notintersect}, fewer than $5c_\eps\gamma \log n$ of them  intersect $x_v$. So, the goal of our next lemma is to show that noise does not disrupt this by too much.

\begin{lemma}\label{lem:Rdecode}
For sufficiently large constant $c_\eps$, with probability at least $1-n^{4-c_\eps\gamma}$, for all nodes $v$, $\tilde R_v = R_v$.
\end{lemma}

\begin{proof}

Conditioning on the event of \Cref{lem:notintersect}, all $C(r)$ for $r\notin R_v$ $c_\eps(c_\eps - 5)\gamma \log n$-intersect $\lnot x_v$. Then, for such an $r$ to be in $\tilde R_v$, more than $\textbf 1(C(r) \land \lnot x_v)-\frac{2\eps+1}{4}c_\eps^2\gamma \log n$ of the intersection positions would have to be flipped by noise. The probability of this is clearly minimized when $\textbf 1(C(r) \land \lnot x_v)$ is as low as possible, i.e., $c_\eps(c_\eps - 5)\gamma \log n$. Then, $c_\eps(c_\eps - 5)\gamma \log n - \frac{2\eps+1}{4}c_\eps^2\gamma \log n = (\frac{3-2\eps}{4}c_\eps - 5 )c_\eps\gamma \log n$ positions must be flipped, and the expected number of such flipped positions is $\mu := \eps (c_\eps - 5)c_\eps\gamma \log n$.

To show a low probability of failure, we need that the number of positions that must be flipped for $r$ to be incorrectly categorized is more than its expectation. To do so, we bound the ratio of the two quantities:

\begin{align*}
	\frac{(\frac{3-2\eps}{4}c_\eps - 5 )c_\eps\gamma \log n}{\eps (c_\eps - 5)c_\eps\gamma \log n} &= \frac{\frac{3-2\eps}{4}c_\eps - 5 }{\eps (c_\eps - 5)}
\\ &\ge \frac{\frac{3-2\eps}{4}c_\eps - 5 }{\frac{c_\eps}{2}} &\text{since $\eps\in(0,\frac12)$}	\\
	& =\frac32-\eps - \frac{10}{c_\eps} 
	\enspace.
\end{align*}

We will set $c_\eps \ge \frac{60}{1-2\eps}$. Then, 
\begin{align*}
	\frac{(\frac{3-2\eps}{4}c_\eps - 5 )c_\eps\gamma \log n}{\eps (c_\eps - 5)c_\eps\gamma \log n} &\ge \frac32-\eps - \frac{1-2\eps}{6} = \frac{4-2\eps}{3}>1\enspace.
\end{align*}

Now that we have bounded the ratio above $1$, we can apply a Chernoff bound:

\begin{align*}
	\Prob{\textbf 1(C(r) \land \lnot \tilde x_v) < \frac{2\eps+1}{4} c_\eps^2\gamma \log n } &\le \Prob{\text{more than $\frac{\frac{3-2\eps}{4}c_\eps - 5}{\eps (c_\eps - 5)}\mu$ intersection positions are flipped}}\\
		&\le exp(- \left(\frac{\frac{3-2\eps}{4}c_\eps - 5}{\eps (c_\eps - 5)}-1 \right)^2 \mu/3 )
		\\
		&\le exp(- \left(\frac{4-2\eps}{3}-1 \right)^2 \mu/3 )\\
		&= exp(- \left(1-2\eps \right)^2 \eps (c_\eps - 5)c_\eps\gamma \log n/27 )
		\enspace.
\end{align*}

We will now further require that $c_\eps \ge \frac{54}{\left(1-2\eps \right)^2\eps}+5$, which gives:

\begin{align*}
	\Prob{\textbf 1(C(r) \land \lnot \tilde x_v) < \frac{2\eps+1}{4} c_\eps^2\gamma \log n } &\le
	exp(- 2c_\eps\gamma \log n) \le n^{-2c_\eps\gamma}
	\enspace.
\end{align*}

Conversely, for some $r'\in R_v$, $C(r)$ does not $1$-intersect $\lnot x_v$ (since it is contained in the superimposition that produces $x_v$). So, for it to $\frac{2\eps+1}{4}c_\eps^2\gamma \log n$-intersect $\lnot \tilde x_v$, at least $\frac{2\eps+1}{4}c_\eps^2\gamma \log n$ of the positions in which $C(r')$ has a \textbf 1 (of which there are exactly $c_\eps^2\gamma \log n$, by definition) would need to be flipped in $\tilde x_v$. The expected number of such flipped positions is $\mu' :=\eps c_\eps^2\gamma \log n$. Since $\eps\in(0,\frac12)$, have $\frac{2\eps+1}{4}c_\eps^2\gamma \log n >\frac{4\eps}{4}c_\eps^2\gamma \log n = \mu'$, so we can again apply a Chernoff bound:

\begin{align*}
	\Prob{\textbf 1(C(r') \land \lnot \tilde x_v) \ge \frac{2\eps+1}{4}c_\eps^2\gamma \log n } &\le \Prob{\text{at least $\frac{2\eps+1}{4}c_\eps^2\gamma \log n$ of $C(r')$'s \textbf 1s are flipped}}\\
	&\le exp(- \left(\frac{\frac{2\eps+1}{4}}{\eps}-1 \right)^2 \mu'/3 ) \\
	&=exp(- \left(\frac{1}{4\eps}-\frac 12 \right)^2 \eps c_\eps^2\gamma \log n/3 )\enspace.
\end{align*}

Requiring that $c_\eps \ge \frac{6}{\eps}\left(\frac{1}{4\eps}-\frac 12 \right)^{-2}$ again gives:

\begin{align*}
	\Prob{\textbf 1(C(r') \land \lnot \tilde x_v) \ge \frac{2\eps+1}{4}c_\eps^2\gamma \log n } &\le
	exp(-2c_\eps\gamma \log n) \le n^{-2c_\eps\gamma}\enspace.
\end{align*}

So, each codeword is correctly placed in or out of $\tilde R_v$ with probability at least $1-n^{-2c_\eps\gamma}$. Taking a union bound over all $2^{c_\eps\gamma\log n}$ codewords, we have $\tilde R_v = R_v$ with probability at least $1-n^{-c_\eps\gamma}$. Finally, taking another bound over all nodes $v\in V$ and removing the conditioning on the event of \Cref{lem:notintersect} (which occurs with probability at least $1- n^{3-c_\eps \gamma}$ ) gives correct decoding at all nodes with probability at least $1-n^{4-c_\eps\gamma}$. The lemma requires setting $c_\eps \ge \max\{\frac{6}{\eps}\left(\frac{1}{4\eps}-\frac 12 \right)^{-2}, \frac{54}{\left(1-2\eps \right)^2\eps}+5,\frac{60}{1-2\eps}\}$.
\end{proof}

We now analyze the second stage of the algorithm, in which nodes transmit their messages using the combined code, and show that this code allows the messages to be decoded.

\begin{lemma}\label{lem:decode}
	In the second phase of the algorithm, with probability at least $1-n^{\gamma+6-c_\eps\gamma}$, all nodes $v$ can successfully decode $\{m_w:w\in N(v)\}$ (so long as $c_\eps$ is at least a sufficiently large constant).
\end{lemma}

\begin{proof}
Conditioned on the event of \Cref{lem:Rdecode}, all nodes $v$ now know $R_v$.	

In the second stage of the algorithm, in the absence of noise $v$ would hear the string $\bigvee_{w\in N(v)} CD(r_w,m_w)$,  which we will denote $y_v$. To decode the message $m_w$, for some $w\in N(v)$, it examines the subsequence $y_{v,w}$ defined by $(y_{v,w})_j = (y_v)_i : \textbf 1_j(C(w)) = i$. We denote the noisy versions of these strings that $v$ actually hears by $\tilde y_v$ and $\tilde y_{v,w}$ respectively. (Note that $v$ does not know which neighbor $w$ the strings $r_w$ and $\tilde y_{v,w}$ belong to, but it can link them together, which is all that is required at this stage.) Node $v$ decodes $m_w$ as the string $\tilde m_w\in \{0,1\}^{\gamma\log n}$ minimizing $d_H(D(\tilde m_w), \tilde y_{v,w})$. We must show that, with high probability, $\tilde m_w = m_w$.

Conditioned on the event of \Cref{lem:notintersect}, each $C(r_w)$ for $w\in N(v)$ does not $5c_\eps\gamma\log n$-intersect $\bigvee_{r_u\in R_v\setminus \{r_w\}} C(r_u)$. That is, there are at least $(c_\eps-5)c_\eps\gamma\log n$ positions in $x_v$ in which $C(r_w)$ has a \textbf 1 and no other $C(r_u)$ for $u\in N(v)$ does. In these positions $j$, $(y_{v,w})_j = D(m_w)_j$. So, in total, \[d_H(D(m_w),y_{v,w}) \le 5c_\eps\gamma \log n\enspace.\]

Under noise, each of the positions in which $y_{v,w}$ matches $D(m_w)$ will be flipped with probability $\eps$ in $\tilde y_{v,w}$. So, denoting $\Exp{d_H(D(m_w),\tilde y_{v,w})}$ by $\mu$, we have: 

\[\eps c_\eps^2 \gamma\log n\le \mu \le \eps c_\eps^2 \gamma\log n+ 5c_\eps\gamma \log n \enspace.\]

Meanwhile,by the property of a $(\gamma \log n,\frac 13)$-distance code, for any $m\ne m_w \in \{0,1\}^{\gamma\log n}$, \[d_H(D(m),y_{v,w}) \ge d_H(D(m),D(m_w))-d_H(D(m_w),y_{v,w}) \ge  \frac 13 c_\eps^2\gamma\log n-   5c_\eps\gamma \log n\enspace.\]

To lower-bound $\Exp{d_H(D(m),\tilde y_{v,w})}$ (which we denote by $\mu'$), we see that $\mu' = (1-\eps)d_H(D(m),y_{v,w}) + \eps( c_\eps^2 \gamma\log n- d_H(D(m),y_{v,w}))$. Since $\eps>\frac 12$, this is minimized when $d_H(D(m),y_{v,w})$ is as small as possible, i.e.,  $\frac 13 c_\eps^2\gamma\log n-   5c_\eps\gamma \log n$. Then, 

\begin{align*}
	\mu' &\ge (1-\eps)(\frac 13 c_\eps^2\gamma\log n-   5c_\eps\gamma \log n) + \eps( c_\eps^2 \gamma\log n- (\frac 13 c_\eps^2\gamma\log n-   5c_\eps\gamma \log n))\\
	&=(\frac 13 c_\eps - \frac 23 \eps c_\eps - 5 + 10\eps + \eps c_\eps)c_\eps\gamma\log n\\
	&\ge \frac{1+\eps}{3}c_\eps^2\gamma\log n -5c_\eps\gamma\log n\enspace.
\end{align*}

Since $\eps<\frac 12$, we have $\frac{1+\eps}{3}>\eps$, and so we can see that for sufficiently large $c_\eps$, $\mu\le \mu'$. So, it remains to show that $d_H(D(m_w),\tilde y_{v,w})$ and $d_H(D(m),\tilde y_{v,w})$ are concentrated around their expectations.

We first show that, with high probability, $d_H(D(m_w),\tilde y_{v,w}) \le \frac{1+4\eps}{6}c_\eps^2\gamma\log n$. Note that if we set $c_\eps\ge \frac{60}{1-2\eps}$, then 
\[\frac{1+4\eps}{6}c_\eps  =\eps c_\eps + \frac{1-2\eps}{6}c_\eps   >  \eps c_\eps+ 5 \enspace,\]
and so 
\[\frac{1+4\eps}{6}c_\eps^2\gamma\log n> \eps c_\eps^2\gamma\log n + 5c_\eps\gamma\log n \ge \mu\enspace.\]

Then, we can apply a Chernoff bound:
\begin{align*}
	\Prob{d_H(D(m_w),\tilde y_{v,w}) \ge \frac{1+4\eps}{6}c_\eps^2\gamma\log n } &\le \mu \cdot \Prob{d_H(D(m_w),\tilde y_{v,w}) \ge \frac{(1+4\eps)c_\eps}{6} \big/\left( \eps c_\eps+ 5\right)}\\
	&\le exp(- \left(\frac{(1+4\eps)c_\eps}{6\eps c_\eps+ 30}-1 \right)^2 \mu/2 ) \\
	&= exp(- \left(\frac{(1-2\eps)c_\eps-30 }{6\eps c_\eps+ 30}\right)^2 \mu/2 ) \enspace.
\end{align*}

The expression $\frac{(1-2\eps)c_\eps-30 }{6\eps c_\eps+ 30}$ is increasing in $c_\eps$. Therefore, if we ensure that $c_\eps\ge \frac{30}{\eps(1-2\eps)}$, we have 
\[\frac{(1-2\eps)c_\eps -30}{6\eps c_\eps+ 30}\ge \frac{\frac{30}{\eps} -30}{\frac{180}{1-2\eps }+ 30} = \frac{\frac{1}{\eps} -1}{\frac{6}{1-2\eps }+ 1} =\frac{(1-\eps)(1-2\eps) }{\eps(7-2\eps)} \enspace.\]

Then,

\begin{align*}
	\Prob{d_H(D(m_w),\tilde y_{v,w}) \ge \frac{1+4\eps}{6}c_\eps^2\gamma\log n }
	&\le exp(- \left(\frac{(1-\eps)(1-2\eps) }{\eps(7-2\eps)}\right)^2 \mu/2 ) \\
	&\le exp(- \left(\frac{(1-\eps)(1-2\eps) }{\eps(7-2\eps)}\right)^2 c_\eps^2\gamma\log n/2 )\enspace.
\end{align*}

Finally, if we also ensure that $c_\eps \ge 6\left(\frac{(1-\eps)(1-2\eps) }{\eps(7-2\eps)}\right)^{-2}$, 

\begin{align*}
	\Prob{d_H(D(m_w),\tilde y_{v,w}) \ge \frac{1+4\eps}{6}c_\eps^2\gamma\log n }
	&\le exp(- 3c_\eps\gamma\log n)\\
	&\le n^{-4c_\eps\gamma}\enspace.
\end{align*}

We similarly wish to show that with high probability, $d_H(D(m),\tilde y_{v,w})> \frac{(1+4\eps)}{6}c_\eps^2\gamma\log n$ (for $m\ne m_w$). Again, since we have set $c_\eps\ge \frac{60}{1-2\eps}$,

\[\frac{(1+4\eps)}{6}c_\eps = \frac{1+\eps}{3}c_\eps - \frac{1-2\eps}{6}c_\eps <  \frac{1+\eps}{3}c_\eps -5\]

So, 

\[\frac{1+4\eps}{6}c_\eps^2\gamma\log n< \frac{1+\eps}{3} c_\eps^2\gamma\log n - 5c_\eps\gamma\log n \le \mu'\enspace.\]
Then, we can apply a Chernoff bound:
\begin{align*}
	\Prob{d_H(D(m),\tilde y_{v,w})\le \frac{(1+4\eps)}{6}c_\eps^2\gamma\log n} &\le \Prob{d_H(D(m),\tilde y_{v,w}) \le \mu'\cdot \frac{(1+4\eps)c_\eps}{6} \big/\left(\frac{1+\eps}{3}c_\eps -5\right)}\\
	&\le exp(- \left(1-\frac{(1+4\eps)c_\eps}{2(1+\eps)c_\eps - 30} \right)^2 \mu/3 ) \\
	&\le exp(- \left(\frac{(1-2\eps)c_\eps - 30 }{6\eps c_\eps+ 30}\right)^2 \mu/3 )\\ 
	&\le  exp(- \left(\frac{(1-\eps)(1-2\eps) }{\eps(7-2\eps)}\right)^2 c_\eps^2\gamma\log n/3 )\\
	&\le   exp(- 2c_\eps\gamma\log n)\\
	&\le n^{-2c_\eps\gamma}
	\enspace.
\end{align*}

Taking a union bound over all strings in $ \{0,1\}^{\gamma\log n}$, we find that with probability at least $1-n^{\gamma-2c_\eps\gamma}$, $d_H(D(m_w),\tilde y_{v,w}) < \frac{1+4\eps}{6}c_\eps^2\gamma\log n$ and $d_H(D(m),\tilde y_{v,w})> \frac{(1+4\eps)}{6}c_\eps^2\gamma\log n$ for all $m\ne m_w$. So, $v$ successfully decodes $m_w$. Another union bound over all $w\in N(v)$ gives probability at least $1-n^{\gamma+1-2c_\eps\gamma}$ that $v$ correctly decodes the entire set $\{m_w:w\in N(v)\}$. Finally, removing the conditioning on the event of \Cref{lem:Rdecode} and taking a further union bound over all nodes $v$, the probability that all nodes correctly decode their neighbors' messages is at least $1-n^{\gamma+6-c_\eps\gamma}$. We required that \[c_\eps \ge \max\left\{\frac{30}{\eps(1-2\eps)}, 6\left(\frac{(1-\eps)(1-2\eps) }{\eps(7-2\eps)}\right)^{-2}\right\} \enspace.\]
\end{proof}

\Cref{lem:decode} shows that \Cref{alg:beepsim} successfully simulates a \bcongest communication round with high probability. By simulating all communication rounds in sequence, we can simulate any $n^{O(1)}$ \bcongest in its entirety at an $O(\Delta\log n)$ overhead. Note that essentially all \bcongest(and \congest) algorithms are $n^{O(1)}$-round, since this is sufficient to inform all nodes of the entire input graph. So the only problems with super-polynomial round complexities would be those in which nodes are given extra input of super-polynomial size. We are not aware of any such problems having been studied, and therefore \Cref{thm:main} applies to all problems of interest.

\begin{theorem}\label{thm:main}
Any $T=n^{O(1)}$-round \bcongest\ algorithm can be simulated in the noisy beeping model in $O(T\Delta \log n)$ rounds, producing the same output with with probability at least $1-n^{-2}$. 	
\end{theorem}

\begin{proof}
Each round of the \bcongest\ algorithm, in which each node $v$ broadcasts a $\gamma\log n$-bit message to all of its neighbors, is simulated using \Cref{alg:beepsim} with sufficiently large constant $c_\eps$. By \Cref{lem:decode}, each simulated communication round succeeds (has all nodes correctly decode the messages of their neighbors) with probability at least $1-n^{\gamma+6-c_eps\gamma}$. Taking a union bound over all $T$ rounds, and choosing $c_\eps$ sufficiently large, gives a probability of at least $1-n^{-2}$ that all simulated communication rounds succeed. In this case, the algorithm runs identically as it does in \bcongest, and produces the same output. The running time of \Cref{alg:beepsim} is $O(\Delta\log n)$, so the overall running time is $O(T\Delta \log n)$.
\end{proof}

We then reach an $O(\Delta^2\log n)$-overhead simulation for \congest.

\begin{corollary}
Any $T=n^{O(1)}$-round \congest\ algorithm can be simulated in the noisy beeping model in $O(T\Delta^2 \log n)$ rounds, producing the same output with with probability at least $1-n^{-2}$. 	
\end{corollary}

\begin{proof}
A $T=n^{O(1)}$-round \congest\ algorithm can be simulated in $O(T\Delta)$ rounds in \bcongest\ as follows: nodes first broadcast their IDs to all neighbors, and then each \congest\ communication round is simulated in $\Delta$ \bcongest\ rounds by having each node $v$ broadcast $\langle ID_u, m_{v\rightarrow u}\rangle$ to its neighbors, for every $u\in N(v)$ in arbitrary order. Then, by \Cref{thm:main}, this algorithm can be simulated in $O(T\Delta^2 \log n)$ rounds.
\end{proof}

\hide{
\section{Devices Without Local Randomness}
We remark that, in the noisy beeping model, we can run our simulation protocols even if nodes to not have inherent access to randomness, since we can extract random bits from the noise behavior. Specifically, we can have all nodes listen for the first $T$ rounds. These nodes then hear $\mathbf 1$ with probability $\eps$ and $\mathbf 0$ with probability $1-\eps$, independently, in each of these rounds. We can use a simple trick attributed to Von Neumann to generate unbiased random bits from this behavior: a node considers these rounds in pairs. If it hears the same result for both rounds in the pair, they are discarded. Otherwise, hearing $\mathbf{01}$ is mapped to to $\mathbf 0$, and $\mathbf{10}$ to $\mathbf 1$.

In this way, each pair of rounds has probability $2\eps(1-\eps)$ of being mapped to a uniform random bit, and is otherwise discarded. Algorithm \ref{alg:beepsim} requires each node to have $\Theta(\log n)$ random bits, and runs for $\Theta(\Delta\log n)$ rounds. So, prior to running Algorithm \ref{alg:beepsim}, we can run $\Theta(\log n)$ rounds of randomness extraction, obtaining sufficient random bits at each node with high probability by a Chernoff bound, without increasing the overall asymptotic running-time.
}

\section{Lower bounds}

We now show lower bounds on the number of rounds necessary to simulate \bcongest and \congest, based on the hardness of a simple problem we call  $B$-bit Local Broadcast. We define the $B$-bit Local Broadcast problem as follows: 

\begin{definition}[$B$-Bit Local Broadcast]
Every node $v$ is equipped with a unique identifier $ID_v\in [n]$. Every node $v$ receives as input $\{\langle ID_u, m_{v\rightarrow u}\rangle : u\in N(v)\}$: that is, a set containing messages $m_{v\rightarrow u}\in\{0,1\}^B$ for each of $v$'s neighbors $u$, coupled with the ID of $u$ to identify the destination node. Each node $v$ must output the set $\{\langle ID_u, m_{u\rightarrow v}\rangle : u\in N(v)\}$ (i.e. the set of messages \emph{from} each of its neighbors, coupled with their IDs).
\end{definition}	
	
\begin{lemma}
$B$-Bit Local Broadcast requires $\Omega(\Delta^2 B)$ rounds in the beeping model (even without noise), for any algorithm succeeding with probability more than $2^{-\frac12\Delta^2 B}$.
\end{lemma}

\begin{proof}
	The graph we use as our hard instance is as follows: we take the complete bipartite graph $K_{\Delta,\Delta}$, and add $n-2\Delta$ isolated vertices. This graph then has $n$ vertices and maximum degree $\Delta$. Arbitrarily fix unique IDs in $[n]$ for each node. We will only consider the nodes of $K_{\Delta,\Delta}$ to show hardness. Arbitrarily denote one part of the bipartition $L$ and the other $R$. For nodes $v\in L$, we choose each $m_{v\rightarrow u}$ independently uniformly at random from $\{0,1\}^{B}$. We set all other $m_{x\rightarrow x}$ to $\textbf 0^{\log n}$ (so, in particular, the inputs for all nodes $u\in R$ are identical).
	
Let $\mathcal R$ denote the concatenated strings of local randomness of all nodes in $R$ (in any arbitrary fixed order). Then, the output of any node $u\in R$ must be fully deterministically dependent on the node IDs (which are fixed), $\mathcal R$, $u$'s input messages (which are identically fixed to be all $\textbf 0$s), and the pattern of beeps and silence of nodes in $L$ (and note that all nodes in $R$ hear the same pattern: a beep if an node in $L$ beeps, and silence otherwise). An algorithm running for $T$ rounds has $2^T$ possible such patterns of beeps and silence.

So, the overall output of all nodes in $R$ must be one of $2^T$ possible distributions, where the distribution is over the randomness of $\mathcal R$. The correct output for these nodes is uniformly distributed over $2^{\Delta^2 B}$ possibilities (the choices of input messages for $L$). The probability of a correct output is therefore at most $2^{T-\Delta^2 B}$. So, any algorithm with $T\le \frac12\Delta^2 B$ succeeds with probability at most $2^{-\frac12\Delta^2 B}$.

\end{proof}

Having shown a lower bound on the problem in the beeping model, upper bounds in \bcongest and \congest imply lower bounds on the overhead of simulation.

\begin{lemma}
$B$-Bit Local Broadcast can be solved deterministically in $O(\Delta \lceil B/\log n\rceil )$ rounds of \bcongest and in $O(\lceil B/\log n\rceil )$ rounds of \congest.
\end{lemma}

\begin{proof}
In \bcongest, each node $v$ simply broadcasts the strings $\langle ID_u, m_{v\rightarrow u}\rangle$ for each $u\in N(v)$, taking $O(\Delta\lceil B/\log n\rceil )$ rounds. In \congest, node $v$ instead sends $m_{v\rightarrow u}$ to node $u$ for each $u\in N(v)$, taking $O(\lceil B/\log n\rceil) $ rounds.
\end{proof}

\begin{corollary}
Any simulation of \bcongest\ in the noiseless beeping model (and therefore also the noisy beeping model) has $\Omega(\Delta\log n)$ overhead. Any simulation of \bcongest\ in the noiseless (and noisy) beeping model has $\Omega(\Delta^2 \log n)$ overhead. 
\end{corollary}

\section{Application: Maximal Matching}

In this section we give an example application of our simulation, to the problem of maximal matching. The problem is as follows: we assume each node has a unique $O(\log n)$-bit ID. For a successful maximal matching, each node must either output the ID of another node, or \textsc{Unmatched}. The outputs must satisfy the following:

\begin{itemize}
	\item Symmetry: iff $v$ outputs $ID(u)$, then $u$ outputs $ID(v)$. Since each node outputs at most one ID, this implies that the output indeed forms a matching.
	\item Maximality: for every edge $\{u,v\}$ in the graph, $u$ and $v$ do not both output \textsc{Unmatched}.
\end{itemize}

To our knowledge, no bespoke maximal matching algorithm has previously been designed for the beeping model (either noisy or noiseless) or for \bcongest. So, the fastest existing beeping algorithm is obtained by simulating the best \congest algorithms using the simulation of \cite{AGL20}. Since an $O(\Delta+\log^* n)$-round \congest algorithm for maximal matching exists \cite{PR01}, the running time under \cite{AGL20}'s simulation is therefore $O(\Delta^4\log n + \Delta^3\log n \log^* n )$.

We show an $O(\log n)$-round \bcongest algorithm for maximal matching, which our simulation then converts to an $O(\Delta\log^2 n)$-round algorithm in the noisy beeping model, thereby improving the running time by around a $\Delta^3/\log n$ factor.

The base of our algorithm is Luby's algorithm for maximal independent set \cite{Luby86}, which can be applied to produce a maximal matching (\Cref{alg:LubyMM}). (Often this algorithm is stated with real sampled $x(e)$ values from $[0,1]$; however, since we must communicate these values using $O(\log n)$-bit messages, we instead use integers from $[n^9]$. It can be seen that with probability at least $1-n^{-4}$, no two of the values sampled during the algorithm are the same, so we can condition on this event for the rest of the analysis and avoid considering ties.)

	\begin{algorithm}[H]
	\caption{Maximal Matching: Luby's Algorithm}
	\label{alg:LubyMM}
	\begin{algorithmic}
		\For{$O(\log n)$ iterations,} 	
		\State Each edge $e$ samples $x(e)$ independently uniformly at random from $[n^9]$
		\State Edge $e$ joins the matching $M$ if $x(e)<x(e')$ for all $e'$ adjacent to $e$
		\State Endpoints of edges in $M$ drop out of the graph
		\EndFor 	
	\end{algorithmic}
\end{algorithm}

It is well-known (see \cite{Luby86}) that Luby's algorithm produces a maximal matching in $O(\log n)$ rounds with high probability. To implement this in \bcongest we must make some minor changes to account for the fact that it is nodes, not edges, that communicate (\Cref{alg:bcongestMM}). 

The aim of the algorithm is as follows: if, in a particular round $i$, an edge $\{u,v\}$ has a lower $x(\{u,v\})$ value than all neighboring edges, the following process occurs. Its higher-ID endpoint (assume W.L.O.G. that this is $u$) first broadcasts $\textsc{Propose}\langle \{u,v\}, x(\{u,v\})\rangle$. The other endpoint $v$ then broadcasts $\textsc{Reply}\langle \{u,v\}\rangle$. Node $u$ broadcasts $\textsc{Confirm}\langle \{u,v\}\rangle$, and finally node $v$ also broadcasts $\textsc{Confirm}\langle \{u,v\}\rangle$. These $\textsc{Confirm}$ messages cause nodes adjacent to $u$ and $v$ to be aware that $u$ and $v$ will be ceasing participation (because they have been matched), and so any edges to them can be discarded from the graph.

\begin{algorithm}[H]
	\caption{Maximal Matching in \bcongest}
	\label{alg:bcongestMM}
	\begin{algorithmic}
		\State Each node $v$ broadcasts its ID
		\State Let $E_v$ be the set of $v$'s adjacent edges
		\State Let $H_v$ be the set of $v$'s adjacent edges for which $v$ the higher-ID endpoint 
		\For{$i = 1$ to $O(\log n)$, in round $i$,}	
		\State $v$ samples $x(e)$ independently uniformly at random from $[n^9]$ for each $e\in H_v$
		\State $v$ broadcasts $\textsc{Propose}\langle e_v, x(e_v)\rangle$, where $x(e_v)$ is the unique minimum of $v$'s sampled values (if it exists)
		\State Let $e'_v$ be the edge with the minimum $x(e'_v)$ for which $v$ received $\textsc{Propose}\langle e'_v, x(e'_v)\rangle$
		\If{$x(e'_v)<x(e_v)$}
		\State $v$ broadcasts $\textsc{Reply}\langle e'_v\rangle$
		\EndIf
		\If{$v$ received $\textsc{Reply}\langle e_v\rangle$ and did not broadcast a $\textsc{Reply}$}
		\State $v$ broadcasts $\textsc{Confirm}\langle e_v\rangle$
		\State $v$ outputs $e_v\in MM$ and ceases participation
		\EndIf
		\If{$v$ received $\textsc{Confirm}\langle e'_v\rangle$}
		\State $v$ broadcasts $\textsc{Confirm}\langle e'_v\rangle$
		\State $v$ outputs $e'_v\in MM$ and ceases participation
		\EndIf	
		\If{$v$ received $\textsc{Confirm}\langle \{w,z\}\rangle$ for any $w,z\ne v$}
		\State $v$ removes $\{w,v\}$ and $\{z,v\}$ from $E_v$ and $H_v$ (if present).
		\EndIf	
		\If{$E_v$ is empty}
		\State $v$ ceases participation
		\EndIf	
		\EndFor
	\end{algorithmic}
\end{algorithm}

\begin{lemma} \label{lem:maxmatch}
	If \Cref{alg:bcongestMM} terminates (i.e. causes all nodes to cease participation), it outputs a maximal matching.
\end{lemma}

\begin{proof}
	We first prove maximality. Nodes only cease participation when they are adjacent to an edge in $MM$, or when they have no remaining adjacent edges. Edges are only removed when they are adjacent to an edge in $MM$. So, upon termination, there are no edges in the original graph that are neither in $MM$ nor adjacent to an edge in $MM$, and therefore $MM$ is a maximal matching.
	
	We now prove independence. Let $\{u,v\}$ be an edge which is added to $MM$ in round $i$, and assume W.L.O.G. that $u$ is the higher-ID endpoint. It is clear that, since $\{u,v\}$ is added to $MM$, we must have the behavior described above ($u$ broadcasts $\textsc{Propose}\langle \{u,v\}, x(\{u,v\})\rangle$, $v$ broadcasts $\textsc{Reply}\langle \{u,v\}\rangle$, $u$ broadcasts $\textsc{Confirm}\langle \{u,v\}\rangle$, $v$ broadcasts $\textsc{Confirm}\langle \{u,v\}\rangle$). Then, we can show that this behavior pattern excludes the possibility that any adjacent edge also joins $MM$ in round $i$:
	
	\begin{enumerate}
		\item $u$ cannot act as the higher-ID endpoint of any other edge joining $MM$, since it only \textsc{Propose}s $\{u,v\}$.
		\item $u$ cannot act as the lower-ID endpoint of any other edge joining $MM$, since it \textsc{Confirm}s an edge it \textsc{Propose}d, and therefore cannot have broadcast any \textsc{Reply}.
		\item $v$ cannot act as the higher-ID endpoint of any other edge joining $MM$, since it broadcasts a \textsc{Reply} and therefore does not $\textsc{Confirm}\langle e_v\rangle$.
		\item $v$ cannot act as the lower-ID endpoint of any other edge joining $MM$, since only broadcasts \textsc{Reply}$\langle \{u,v\}\rangle$, and does not \textsc{Reply} for any other edge.
	\end{enumerate}
	
	So, no adjacent edge to $\{u,v\}$ can join in round $i$. Furthermore, all nodes adjacent to $u$ and $v$ receive a $\textsc{Confirm}\langle \{u,v\}\rangle$ message and therefore all other edges adjacent to $u$ and $v$ are removed from the graph. So, no edge adjacent to $\{u,v\}$ can be added to $MM$ in future rounds either. This guarantees that $MM$ is an independent set of edges.
\end{proof}

\begin{notation}
	We will use the notation $e\sim e'$ to mean $e\cap e' \ne \emptyset$, i.e., $e'$ shares \emph{at least} one endpoint with $e$ (and can be $e$ itself). We will denote $|\{e'\in E: e\sim e'\}|$ by $d(e)$, i.e. the number of adjacent edges of $e$, including $e$ itself.
\end{notation}

\begin{lemma}
	In any particular round $i$, the expected number of edges removed from the graph is at least $\frac m2 $.
\end{lemma}
This lemma refers to the \emph{current} graph at round $i$, i.e. without all edges and nodes that have been removed in previous rounds, and $m$ is accordingly the number of edges in the current graph.

\begin{proof}
	It is easy to see that, as intended, an edge $\{u,v\}$ is added to $MM$ if it has a lower $x(\{u,v\})$ value than all neighboring edges: its higher-ID endpoint (W.L.O.G. $u$), which sampled the value $x(\{u,v\})$, will denote the edge as $e_u$ and \textsc{Propose} it, the $x(\{u,v\})$ value will be lower than that for which $v$ \textsc{Propose}d and so $v$ will $\textsc{Reply}\langle \{u,v\}\rangle$, and $u$ will not hear a lower-valued edge to \textsc{Reply} and will therefore $\textsc{Confirm}\langle \{u,v\}\rangle$. $v$ will also $\textsc{Confirm}\langle \{u,v\}\rangle$, and all edges adjacent to $\{u,v\}$ will be removed from the graph. There are $d(u)+d(v)-1$ such edges.
	
	The probability that $x(\{u,v\})< x(e)$ for all $e\sim\{u,v\}$ with $e\ne \{u,v\}$ is $\frac{1}{d(u)+d(v)-1}$. So, the expected number of edges removed from the graph is at least $\frac{1}{2}\sum_{\{u,v\}\in E}  (d(u)+d(v)-1)\frac{1}{d(u)+d(v)-1} = \frac m2$ (where the $\frac{1}{2}$ factor arises since each end can be removed by either of its endpoints being matched, so is double-counted in the sum).
\end{proof}

\begin{lemma}\label{lem:bcongestmatching}
	\Cref{alg:bcongestMM} performs maximal matching in $O(\log n)$ rounds of \bcongest, succeeding with high probability.
\end{lemma}

\begin{proof}
	By Lemma \ref{lem:maxmatch}, Algorithm \ref{alg:bcongestMM} produces a maximal matching if it terminates. Conditioning on the event that all sampled values are distinct, the algorithm removes at least half of the edges in the graph in each iteration in expectation. After $4\log n$ iterations, therefore, the expected number of edges remaining is at most $n^2 \cdot n^{-4} = n^{-2}$, and therefore by Markov's inequaility, with probability at least $1-n^{-2}$ the number of edges remaining is $0$ and the algorithm has terminated. Removing the conditioning on the event that sampled values are distinct, the algorithm terminates with probability at least $1-n^{-2}-n^{-4}$.
\end{proof}

\begin{theorem}
Maximal matching can be performed in $O(\Delta \log^2 n)$ rounds in the noisy beeping model, succeeding with high probability.
\end{theorem}

\begin{proof}
Follows from applying Theorem \ref{thm:main} to Lemma \ref{lem:bcongestmatching} .
\end{proof}

This is close to optimal, since we show an $\Omega(\Delta \log n)$ bound even in the noiseless model:

\begin{theorem}\label{thm:matchinglower}
Maximal matching requires $\Omega(\Delta \log n)$ rounds in the (noiseless) beeping model, to succeed with any constant probability.
\end{theorem}

\begin{proof}
Our hard ensemble of instances is as follows: the underlying graph will be $K_{\Delta,\Delta}$, the complete bipartite graph with $\Delta$ vertices in each part. Each node's ID will be drawn independently at random from $[n^4]$.

Arbitrarily naming the two parts of the graph left and right, we consider the outputs of nodes on the right. For a correct output to maximal matching, each node on the right must uniquely output the ID of a node on the left, and so the union of outputs of the right part must be the list of IDs of the right part. The number of possible such lists (even assuming that IDs are all unique and the IDs of the right side are fixed) is $\binom{n^4-\Delta}{\Delta}\ge \binom{\frac 12 n^4}{\Delta} \ge \left( \frac{n^4}{2\Delta}\right)^\Delta\ge n^{3\Delta}$.

We note that each right node's output must be dependent only on its ID, its local randomness, and the transcript of communication performed by left nodes during the course of the algorithm. Since the graph is a complete bipartite graph, in each round there are only two discernable possibilities for communication from the perspective of right-part nodes: either at least one left node beeps, or none do. So, the transcript for an $r$-round algorithm can be represented as a sequence $\{B,S\}^r$, corresponding to hearing a beep or silence in each round. There are $2^r$ such transcripts.

Therefore, the union of output from right nodes depends solely on the randomness of right nodes, the IDs of left nodes, and the transcript. Of these, only the transcript can depend on left nodes' IDs. Each transcript therefore induces a distribution of right-part outputs (over the randomness of right-side IDs and local randomness.

There must be some set of left-part IDs such that under any transcript, the probability that the right-side nodes correctly output that set is at most $2^r / n^{3\Delta}$. So, if $r\le \Delta \log n$, then the probability that the right part produces a correct output on this instance is at most $n^\Delta / n^{3\Delta} =n^{-2\Delta} = o(1)$.
\end{proof}

\section{Conclusions}

We have presented an optimal method for simulating \bcongest and \congest in the noisy (and noiseless) beeping model. We have also presented, as an example, a maximal matching algorithm which requires $O(\log n)$ rounds in \bcongest, and which, using our simulation, can therefore be run in $O(\Delta \log^2 n)$ rounds in the noisy beeping model.

While our general simulation method is optimal, there is still room for improvement for many specific problems in the beeping model, and the complexity picture has significant differences from the better-understood message passing models. For example, in \congest, the problems of maximal matching and maximal independent set have similar $O(\log \Delta + \log^{O(1)}\log n)$ randomized round complexity upper bounds \cite{BEPS16,F20,G16, RG20}, whereas in the beeping model, maximal independent set can be solved in $\log^{O(1)} n$ rounds \cite{AABCHK13} while maximal matching requires $\Omega(\Delta \log n)$ (\Cref{thm:matchinglower}). In general, the question of which problems can be solved in $O(\log^{O(1)} n)$ rounds in the beeping model, and which require $poly(\Delta)$ factors, remains mostly open.

\bibliographystyle{plain}

\bibliography{beepsim}
\appendix

\end{document}